\documentclass{article}
\usepackage{amsmath, amssymb, amsthm}
\usepackage{graphicx}
\usepackage{euler} % math %\usepackage{eulervm} % a better implementation of the euler package (not in gwTeX)
\usepackage{pstricks}
\usepackage{authblk}

\def\margin{2.45cm}
\usepackage[left=\margin,right=\margin,top=\margin,bottom=\margin]{geometry}
\title{Another Disjoint Compression Algorithm for OCT}
\author{R. Krithika}
\author{N. S. Narayanaswamy}
\affil{Department of Computer Science and Engineering, \authorcr Indian Institute of Technology Madras, India.\authorcr
\{krithika $\mid$ swamy\}@cse.iitm.ac.in}
\date{}
\theoremstyle{plain}
\newtheorem{theorem}{Theorem}
\newtheorem{lemma}[theorem]{Lemma}
\newtheorem{cor}[theorem]{Corollary}
\newtheorem{obs}[theorem]{Observation}
\begin{document}
\maketitle
\begin{abstract}
We describe an elegant $O^*(2^k)$ algorithm for the disjoint compression problem for Odd Cycle Transversal based on a reduction to Above Guarantee Vertex Cover. We believe that this algorithm refines the understanding of the Odd Cycle Transversal algorithm by Reed, Smith and Vetta \cite{oct-fpt}.
\end{abstract}
\section{Introduction}
Given an undirected graph, an odd cycle transversal ({\footnotesize OCT}) is a subset of vertices whose deletion makes the resulting graph bipartite. The natural optimization problem is to find a minimum cardinality such set and the corresponding decision problem is known to be NP-complete \cite{garey}. We revisit the parameterized version of this problem that is known to be fixed-parameter tractable (FPT) with respect to the solution size as the parameter \cite{oct-fpt}.
\begin{center}
\noindent \fbox{
  \parbox{11.7cm}{
\noindent {\footnotesize{ODD CYCLE TRANSVERSAL}}\\
\textbf{Input: }A graph $G$ and a non-negative integer $k$\\
\textbf{Parameter: }$k$\\
\textbf{Question:} Does there exist $S \subseteq V(G)$, $|S| \leq k$ such that $G - S$ is bipartite?
}
}
\end{center} 
\indent Parameterized algorithm analysis is a multi-dimensional analysis of the running time as a function of the input size and parameter(s). A decision problem with input size $n$ and a parameter $k$ is said to be FPT if it admits an algorithm with runtime $f(k)n^{O(1)}$. Such an algorithm is referred to as an FPT algorithm for the problem and the corresponding runtime is called as an FPT runtime. The running time $f(k)n^{O(1)}$ of an FPT algorithm is generally denoted as $O^*(f(k))$ by suppressing the polynomial terms.\\ 
\indent The fixed-parameter tractability of {\footnotesize OCT} was first shown in \cite{oct-fpt} by an $O^*(3^k)$ algorithm. This was obtained through a technique which is now called `iterative compression', that led to the design of FPT algorithms for many other problems. This technique typically works for minimization problems parameterized by the solution size. The idea is to begin with a solution of size $k+1$ and attempt to compress it (in FPT time) to a solution of size $k$. This is the compression step. To get the $k+1$ sized solution to start with, an algorithm using iterative compression technique typically starts with a $k+1$ sized solution for an induced subgraph on (any) $k+2$ vertices and tries to compress the solution to one of size $k$. If it succeeds, it iteratively adds a new vertex both to the graph and to the solution and continues the compression step in the (larger) new graph. This process continues until we reach the original graph or we get a no answer for any intermediate induced subgraph. If a $k$ sized solution exists for $G$, we are guaranteed to find it within $n-k$ compressions. \\
\indent Typically, algorithms employing iterative compression (including the one in \cite{oct-fpt}), the interaction of the known $k+1$ sized solution $S$ with a smaller solution (if one exists) is exploited to obtain a $k$ sized solution. So, part of the FPT time incurred in the compression step is generally due to the subset enumeration of $S$ as possible choices for the intersection of $S$ with the $k$ sized solution that we seek for. For each such subset $U$, the compression subtask is essentially to find a solution of $G-U$ (for which $S \setminus U$ is a solution) that is disjoint from $S \setminus U$. This task is called as the {\em disjoint compression} step. The disjoint compression problem for {\footnotesize OCT} is defined as follows.
\begin{center} 
%\vspace{.15cm}
\noindent \fbox{
  \parbox{11.7cm}{
\noindent {\footnotesize{OCT-DISJOINT-COMPRESSION}}\\
\textbf{Input: }A graph $G$ and an {\footnotesize OCT} $T$ of $G$ such that $G[T]$ is bipartite\\
\textbf{Question:} Does $G$ have an {\footnotesize OCT} $T'$ of size at most $|T|-1$ such that $T \cap T'=\emptyset$?
}
}
\end{center} 
\noindent {\footnotesize OCT-DISJOINT-COMPRESSION} is known to be NP-complete \cite{disjoint-comp}. The {\footnotesize OCT} algorithm in \cite{oct-fpt} solves this problem in $O^*(2^k)$ time by showing that the smaller sized disjoint solution (if one exists) can be obtained as a separator in a graph among a set of auxiliary graphs, each of which can be constructed in polynomial time \cite{oct-fpt}. An alternate $O^*(2^k)$ algorithm for {\footnotesize OCT-DISJOINT-COMPRESSION} is presented in \cite{oct-color} which led to another $O^*(3^k)$ algorithm for {\footnotesize OCT}. A third $O^*(3^k)$ algorithm for {\footnotesize OCT} is described in \cite{oct-separator}. This algorithm solves {\footnotesize OCT} by solving $O^*(3^k)$ instances of a variant of disjoint compression rather than solving instances of {\footnotesize OCT-DISJOINT-COMPRESSION}. All these known compression based algorithms essentially transform the problem of finding an {\footnotesize OCT} to a vertex separator question in FPT time. After nearly a decade, the $O^*(3^k)$ bound for {\footnotesize OCT} was improved to $O^*(2.3146^k)$ \cite{fpt-lp}. This runtime was achieved by a branching algorithm that employs linear programming techniques and a reduction from {\footnotesize OCT} to the well-known Above Guarantee Vertex Cover ({\footnotesize AGVC}) problem. 
\begin{center} 
\noindent \fbox{
  \parbox{11.7cm}{
\noindent {\footnotesize{ABOVE GUARANTEE VERTEX COVER}}\\
\textbf{Input: }A graph $G$, a maximum matching $M$ and a non-negative integer $k$\\
\textbf{Parameter: }$k$\\
\textbf{Question:} Does $G$ have a vertex cover of size at most $|M|+k$?%he above guarantee parameter.\\
}
}
\end{center} 
\noindent The algorithm in \cite{fpt-lp} exploits the structure of the vertex cover polytope and the parameterized equivalence between {\footnotesize OCT} and {\footnotesize AGVC} to achieve the improved runtime. \\
\indent In this work, we combine the iterative compression technique and the reduction from {\footnotesize OCT} to {\footnotesize AGVC} to describe a conceptually simpler $O^*(2^k)$ algorithm for {\footnotesize OCT-DISJOINT-COMPRESSION}. As opposed to the vertex separator subroutine that the known compression algorithms for {\footnotesize OCT} employ, we transform (in FPT time) the {\footnotesize OCT} question in the compression step to the vertex cover problem in (multiple) bipartite graphs. As a consequence, we obtain yet another $O^*(3^k)$ algorithm for {\footnotesize OCT}. %Parameterized generalizations of {\footnotesize OCT} and {\footnotesize AGVC} are studied and the inter-relationship among them are described in \cite{r-part,saket}. This work also adds to the understanding of this 
\section{OCT via AGVC}  
We describe an $O^*(2^k)$ algorithm for {\footnotesize OCT-DISJOINT-COMPRESSION} by transforming the {\footnotesize OCT} instance in the compression step to an {\footnotesize{AGVC}} instance. The transformation is described below.\\
{\bf OCT Reduces to AGVC \cite{fpt-lp}: }Given a graph $G$ on $n$ vertices, we construct the graph, denoted by $G^2$, on the vertex set $V_1 \cup V_2$ where $V_i=\{v_i \mid v \in V(G)\}$ for $i \in \{1,2\}$. The edge set of $G^2$ is $\{\{u_1,v_1\},\{u_2,v_2\} \mid \{u,v\} \in E(G)\} \cup \{\{v_1,v_2\} \mid v \in V(G)\}$. For a set $S$ of vertices of $G$, let $S_i$ denote the set $\{v_i \in V_i \mid v \in S\}$ of vertices in $G^2$. 
\begin{lemma}\cite{fpt-lp}%We 
\label{reduc}
$G$ has an {\footnotesize OCT} $S$ of size $k$ if and only if $G^2$ has a vertex cover $X$ of size $n+k$. 
\end{lemma}
\begin{cor}
\label{cor-reduc}
\cite{fpt-lp}
For an {\footnotesize OCT} $S$ of $G$, if $P \uplus Q$ is a bipartition of $G - S$, then $P_1 \cup Q_2$ and $P_2 \cup Q_1$ are independent sets in $G^2$. That is, $V(G^2)\setminus (P_1 \cup Q_2)$ and $V(G^2)\setminus (P_2 \cup Q_1)$ are vertex covers of $G^2$. Conversely, if $I$ is an independent set in $G^2$, then the sets $P=\{v \in V(G) \mid v_1 \in I\}$ and $Q=\{v \in V(G) \mid v_2 \in I\}$ form a bipartition of $G - S$ where $S=\{v \in V(G) \mid v_1,v_2 \in V(G^2) \setminus I\}$ is an OCT of $G$. 
\end{cor}
\noindent {\bf The Compression Step: }Let $G$ be a graph on $n$ vertices and $S$ be an $k+1$ sized {\footnotesize OCT} of $G$ in the compression step. Let $T$ be a subset of $S$ that induces a bipartite graph. Let $H$ denote the subgraph of $G$ induced on $V(G) \setminus (S \setminus T)$. Let $|V(H)|=h$ and $B$ denote the set $V(H) \setminus T$. Note that $H[B]$ is bipartite since $V(H) \setminus T$ is $V(G) \setminus S$ and hence $T$ is an {\footnotesize OCT} of $H$. Now, the task is to determine if $H$ has an {\footnotesize OCT} of size at most $|T|-1$ that is disjoint from $T$. From Lemma \ref{reduc}, we have the following observation.
\begin{obs}
\label{main-thm}
$H$ has an {\footnotesize OCT} $T' \subseteq B$ (in other words, $T'$ is disjoint from $T$) of size $r$ if and only if $H^2$ has a vertex cover $X$ of size $h+r$ such that for each $v \in T$, either $v_1 \in X$ or $v_2 \in X$ but not both.
\end{obs}
\noindent We now describe an algorithm for {\footnotesize OCT-DISJOINT-COMPRESSION} based on Observation \ref{main-thm}.
%Thus, searching for an {\footnotesize OCT} in $H$ that excludes $T$ is equivalent to searching for a set $X$ in the appropriately restricted space of vertex covers of $H^2$. This search is described in Algorithm {\em OCT-Disjoint-Compression$(H,T)$}.

%\vspace{.1cm}
\noindent \fbox{
  \parbox{16cm}{
{\bf Algorithm} {\em OCT-Disjoint-Compression$(H,T)$} {\tt{/* Disjoint Compression Step of {\footnotesize OCT} */}}\\
{\em Input:} A graph $H$ and an {\footnotesize OCT} $T$ of $H$\\
{\em Output:} An {\footnotesize OCT} $T'$ of $H$ of size at most $|T|-1$ such that $T \cap T'=\emptyset$ (if one exists)\\
{\bf (I) }Iterate over each set $Y \subset (T_1 \cup T_2)$ such that for each $v \in T$, either $v_1 \in Y$ or $v_2 \in Y$ but not both.\\
\indent \hspace{.5cm}{\tt{/* $2^{|T|}$ choices for $Y$ */}}\\
\indent \hspace{.5cm}{\bf 1. }If $Y$ is not a vertex cover of $H^2[T_1 \cup T_2]$, then skip to the next choice of $Y$.\\
\indent \hspace{.5cm}{\bf 2. }Define $W \subseteq (B_1 \cup B_2)$ as the set of vertices that are adjacent to some vertex in $(T_1 \cup T_2)\setminus Y$.\\
\indent \hspace{1cm}Obtain a minimum vertex cover $Z$ of $H^2[(B_1 \cup B_2) \setminus W]$.\\
\indent \hspace{1cm}Define $X$ as $Y \cup W \cup Z$. {\tt{/* $X$ is a min vertex cover of $H^2$ such that $X \cap (T_1 \cup T_2)=Y$ */}}\\
\indent \hspace{.5cm}{\bf 3. }Define $T'$ as the set $\{v \in V(H) \mid v_1,v_2 \in X\}$. If $|T'| \leq |T|-1$, then return $T'$. \\
{\bf (II) }Declare that $H$ has no {\footnotesize OCT} $T'$ of at most $|T|-1$ vertices that is disjoint from $T$.
}
}
\begin{theorem}
Algorithm OCT-Disjoint-Compression$(H,T)$ determines in $O^*(2^{|T|})$ time whether $H$ has an {\footnotesize OCT} $T'$ of size at most $|T|-1$ such that $T \cap T'=\emptyset$.
\end{theorem}
\begin{proof}%Note that $H^2-(U_1 \cup U_2)$ is bipartite since $H-T'$ is bipartite. That is,  Now, 
Let $T'$ be an {\footnotesize OCT} of size $t$ ($t \leq |T|-1$) in $H$ such that $T \cap T'=\emptyset$. By Observation \ref{main-thm} and Corollary \ref{cor-reduc}, $H^2$ has a vertex cover $X'$ of size $h+t$ such that for each vertex $v$ in $T'$, both $v_1$ and $v_2$ are in $X'$ and for each $v \in V(H)\setminus T'$, either $v_1 \in X'$ or $v_2 \in X'$ but not both. As $T \cap T'=\emptyset$, it follows that $T \subseteq V(H)\setminus T'$. Thus, $X'$ is a vertex cover of $H^2$ such that for each $v \in T$, either $v_1 \in X'$ or $v_2 \in X'$ but not both. Let $Y=X' \cap (T_1 \cup T_2)$. The enumeration of $2^{|T|}$ sets in step (I) of the algorithm will enumerate $Y$. Hence, the sets $W$ and $Z$ in step 2 of the algorithm together with $Y$ will lead to the discovery of a vertex cover $X$ of $H^2$ that is not larger than $X'$. Moreover, $X$ satisfies the property that for each $v \in T$, either $v_1 \in X$ or $v_2 \in X$ but not both. Thus, $\{v \in V(H) \mid v_1,v_2 \in X\}$ is an OCT of $H$ that is disjoint from $T$ and not larger than $T'$. On the other hand, if $H$ has no {\footnotesize OCT} of size at most $|T|-1$ that is disjoint from $T$, then by Observation \ref{main-thm}, every vertex cover $X$ in $H^2$ satisfying the property that for each $v \in T$, either $v_1 \in X$ or $v_2 \in X$ but not both, is of size at least $h+|T|$. Thus, the algorithm will exit from step (II) reporting the failure of the disjoint compression step.\\
For a fixed choice of $Y$, the subsequent steps of the algorithm require only a polynomial-time effort because $W$ is uniquely determined and $Z$ is obtained in polynomial time \cite{west} since the subgraph $H^2[B_1 \cup B_2]$ is bipartite (as $H[B]$ is bipartite). Therefore, the runtime of the algorithm is bounded by $O^*(2^{|T|})$.
\end{proof}
\noindent If OCT-Disjoint-Compression$(H,T)$ returns an {\footnotesize OCT} $T'$, then $T' \cup (S \setminus T)$ is an $k$-sized {\footnotesize OCT} of $G$. Otherwise, we proceed to the next choice of $T$. Therefore, by trying all possible subsets of $S$, we determine if there exists an {\footnotesize OCT} of size at most $k$ in $G$. If no subset of $S$ yields the required {\footnotesize OCT} then we declare that $G$ has no {\footnotesize OCT} of $k$ vertices. Thus, we obtain an algorithm for {\footnotesize OCT} with $\sum_{i=1}^{k+1} {k+1 \choose i}2^{i}=O^*(3^{k})$ runtime.
%\vspace{-.2cm}
\bibliographystyle{alpha}
\bibliography{oct-ref}
\end{document}